\documentclass[10pt, conference]{IEEEtran}
\IEEEoverridecommandlockouts
\usepackage{cite}
\usepackage{enumitem}
\usepackage{amsmath,amssymb,amsfonts}
\usepackage{algorithmic}
\usepackage{graphicx}
\usepackage{textcomp}
\usepackage{xcolor}
\usepackage{amsmath}
\usepackage{amsfonts}
\usepackage{xcolor}
\usepackage{subcaption}
\usepackage{float}
\usepackage{amsthm}
\usepackage{mathtools}
\usepackage{graphicx}
\usepackage{bbm}
\usepackage{hyperref}

\newcommand{\E}{\mathbb{E}}

\DeclareMathOperator{\BDC}{\mathrm{BDC}}

\DeclareMathOperator*{\argmax}{arg\,max}

\DeclareMathOperator{\BSC}{\mathrm{BSC}}

\newcommand{\CC}{\mathcal{C}}
\newcommand{\GG}{\mathcal{G}}

\newcommand{\DD}{\mathcal{D}}

\newcommand{\N}{\mathbb{N}}

\newcommand{\pr}{\mathbb{P}}

\newcommand{\XX}{\mathcal{X}}
\newcommand{\YY}{\mathcal{Y}}

\newcommand{\R}{\mathbb{R}}

\newcommand{\til}[1]{\widetilde{#1}}

\newcommand{\Aut}{\mathsf{Aut}}
\newcommand{\Perm}{\mathsf{Perm}}
\newcommand{\Ne}{\mathsf{ne}}
\newcommand{\supp}{\mathsf{supp}}

\newcommand{\eqdist}{\stackrel{\DD}{=}}

\newcommand{\injectsto}{\xhookrightarrow{}}

\newcommand{\quot}[2]{(\mathclose{#1}/\mathopen{#2})}

\theoremstyle{definition}

\newtheorem{thm}{Theorem}
\newtheorem{lem}[thm]{Lemma}
\newtheorem{cor}[thm]{Corollary}

\newtheorem{conj}[thm]{Conjecture}
\newtheorem{rem}[thm]{Remark}
\newtheorem{prop}[thm]{Proposition}

\newtheorem{defn}[thm]{Definition}

\numberwithin{thm}{section}
\def\BibTeX{{\rm B\kern-.05em{\sc i\kern-.025em b}\kern-.08em
    T\kern-.1667em\lower.7ex\hbox{E}\kern-.125emX}}

\DeclareMathOperator{\Ch}{\mathsf{Ch}}

\begin{document}

\title{On the Symmetries of the Deletion Channel \thanks{Work partially supported by CURIS 2021.}}

\author{\IEEEauthorblockN{Francisco Pernice}
\IEEEauthorblockA{\textit{Stanford University} \\
fpernice@stanford.edu}
}

\maketitle

\begin{abstract}
In this paper, we consider a class of symmetry groups associated to communication channels, which can informally be viewed as the transformations of the set of inputs that ``commute'' with the action of the channel. These groups were first studied by Polyanskiy in \cite{polyanskiy2012}. We show the simple result that the input distribution that attains the maximum mutual information for a given channel is a ``fixed point'' of its group. We conjecture (and give empirical evidence) that the channel group of the deletion channel is extremely small (it contains a number of elements constant in the blocklength). We prove a special case of this conjecture. This serves as some formal justification for why the analysis of the binary deletion channel has proved much more difficult than its memoryless counterparts. 
\end{abstract}

\begin{IEEEkeywords}
channel symmetries, groups, deletion channel.
\end{IEEEkeywords}

\section{Introduction}
Many natural models of communication errors, like those captured by the class of discrete memoryless channels, are by now well understood. Their capacity has been known since Shannon's original paper \cite{shannon}, and codes with efficient encoding and decoding algorithms have been proved to achieve the capacity (e.g., \cite{arikan}). By contrast, other similarly natural error models, like those captured by the \emph{binary deletion channel} or other synchronization channels, are much less well understood. For example, the capacity of the deletion channel is unknown, although several lower and upper bounds have been proved \cite{hard-lower-bound, simple-lower-bound-1/9,Mahdi-upper-bounds, mahdi-review}. In this paper, we work towards an answer to the following question: can we give formal justification for why the binary deletion channel is so much more difficult to analyze than, say, the binary symmetric channel? In a talk in 2008, Mitzenmacher \cite{mitzenmacher-talk} gave the following example to illustrate the difference between the two. Consider the strings 
\[
s_1 = 00000,\qquad s_2=01010.
\]
From the point of view of the binary deletion channel, these two strings seem quite different: for example, deleting any one bit from $s_1$ will produce the same output, while deleting different single bits from $s_2$ will produce all different outputs. On the other hand, from the point of view of the binary symmetric channel, these two strings seem ``equivalent,'' in the sense that there is no clear formal way to distinguish them in terms of the consequences that bit flips have on them. Mitzenmacher went on to say that ``erasure and error channels have pleasant symmetries;  deletion channels do not'' and that ``understanding this asymmetry seems fundamental'' \cite{mitzenmacher-talk}.

In this paper, we consider a symmetry group $\GG_{\Ch}$ associated to any given communication channel $\Ch$, which can informally be viewed as the set of transformations of the set of inputs that ``commute'' with the action of the channel. These groups were first studied in the context of nonasymptotic coding converse bounds by Polyanskiy \cite{polyanskiy2012}, but as far as we know, have never been applied to the deletion channel. For a general class of channels, we show the simple result that the distribution over the inputs that maximizes the mutual information between the input and output of a given channel is a ``fixed point'' of the action of the channel's group. This is analogous to Polyanskiy's results \cite{polyanskiy2012} on the invariance under these group actions of the distributions which solve minimax problems at the heart of finite-length channel coding converse results \cite{finite-length}. These invariance theorems motivate the study of these groups as a coarse measure of the ``hardness'' of a channel: for channels with large symmetry groups, one can vastly simplify capacity calculations, while for asymmetric channels, one gets no such assistance.

We compare the cases of the binary symmetric channel (BSC) and binary deletion channel (BDC). The case of the BSC was already studied by Polyanskiy \cite{polyanskiy2012}; we mainly treat it here to introduce our notation and for contrast with the case of the BDC, which is our main focus. We conjecture (and give empirical evidence) that the channel group of the  BDC is extremely small (it contains a number of elements constant in $n$); this is in contrast with memoryless channels, whose group size was shown by Polyanskiy \cite{polyanskiy2012} to grow at least like $n!$ (by virtue of a natural inclusion $\iota :S_n\injectsto \GG_{\Ch}$ for $\Ch$ memoryless). We prove a special case of our conjecture: within the class of symmetries given by permutations of the indices, $\GG_{\BDC}$ has only two elements.

Given the channel group, one can define the natural induced notion of equivalence between strings: we say strings $s_1$ and $s_2$ are equivalent if they lie on the same orbit of the channel group action on $\{0,1\}^n$, i.e. if there exists a group element mapping $s_1$ to $s_2.$ In this formal sense, the two example strings given by Mitzenmacher are equivalent with respect to the binary symmetric channel, but not with respect to the binary deletion channel. More generally, given any code $\CC$ for a channel, applying any group element to all its codewords yields a new code $\CC'$ that is ``equivalent'' to $\CC$ in a formal sense: $\CC$ and $\CC'$ have the same number of codewords, and the existence of a decoder for $\CC$ implies the existence of a decoder for $\CC'$ with the same error probability, and vice versa.


\subsection{Correction} 
A prior version of this manuscript incorrectly claimed that channel groups were first defined in this work. We thank anonymous reviewers for their helpful comments, which led the author to the work of Polyanskiy \cite{polyanskiy2012} where, as discussed above, channel groups had already been defined and studied in a different context. Moreover, the prior version contained a theorem regarding the uniqueness of the channel group for channel families forming homogeneous Markov chains; the proof had a bug, so the result has been removed.

\section{Definitions}
\subsection{Notation and Elementary Definitions}
Throughout this paper, for $\XX$ a set (alphabet), $\XX^n$ denotes the set of strings of $n$ symbols from $\XX;$ we also let $\XX^* = \bigcup_{j=0}^\infty \XX^j$ and $\XX^{\leq n} = \bigcup_{j=0}^n \XX^j$. For $x\in \XX^n,$ we let $x_i^j \in \XX^{j-i+1}$ denote the substring of $x$ starting at index $i$ and ending at $j,$ inclusive; unless otherwise specified, we let $x_i:=x_i^i$. If $\XX$ can naturally be viewed as a field, we use $x\in \XX^n$ to refer to the vector space element or the string interchangeably. For $x,y\in \XX^*,$ we let $xy$ denote their concatenation. For $x\in \XX^n,$ we let $|x|=n$ denote the string length. All logs (hence entropies, etc.) in this paper are of base equal to the alphabet size unless otherwise specified.

To treat symmetry groups of general channels, it will be useful to view a channel as acting on strings of arbitrary length.
\begin{defn}
For $\Omega$ a probability space and $\XX, \YY$ sets (alphabets),\footnote{In this paper, we take all alphabets to be finite.} a channel is a map $\Ch: \XX^*\times \Omega\to \YY^*.$ For $x\in \XX^*,$ we write $\Ch x$ for the random variable $\omega\mapsto \Ch(x,\omega).$
\end{defn}
For completeness, we give definitions of memoryless channels, the binary symmetric channel and the binary deletion channel in this notation.
\begin{defn}\label{defn:memoryless-channels}
A channel $\Ch:\XX^*\times \Omega \to \YY^*$ is \emph{memoryless} if $x\in \XX$ implies $\Ch x \in \YY$ with probability 1, and for $x\in \XX^n$ we have $\Ch x \eqdist (\Ch^1 x_1)\dots(\Ch^n x_n),$ where $\eqdist$ denotes equality in distribution and the $\Ch^i$ are independent copies of $\Ch.$
\end{defn}
\begin{defn}
Let $\XX = \YY=\{0,1\}, p\in [0,1],$ and $\Omega = \{0,1\}^\infty$ (the infinite product space) with a $Bernoulli(p)^\infty$ measure (the infinite product measure). The \emph{binary symmetric channel} acts on an input $x\in \{0,1\}^n$ as $\BSC_p(x,\omega)= x + \omega_1^n,$ where addition is elementwise and mod 2. The \emph{binary deletion channel} acts on an input $x\in \{0,1\}^n$ as $\BDC_p(x,\omega) = x_{i_1}x_{i_2}\dots x_{i_k},$ where $|x| - k$ is the hamming weight of (number of ones in) $\omega_1^n$, and $i_j$ is the index of the $j$th zero in $\omega.$
\end{defn}

It will sometimes be useful to consider the action of the channel only on strings of a particular length.
\begin{defn}
Let $\Ch:\XX^*\times \Omega\to \YY^*$ be a channel, $n\in\N,$ and let $\Ch|_n:\XX^n\times \Omega\to\YY^*$ be the restriction of $\Ch$ to the strings of length $n$. Suppose there exists $k=k(n)\in \N$ such that the image of $\Ch|_n$ is contained in $\YY^{\leq k}.$ In that case let $m=m(n)$ the minimal such $k$. Then the \emph{$n$th transition matrix of $\Ch$} is the linear map $M_n:\R^{\XX^n}\to \R^{\YY^{\leq m}}$ giving the transition probabilities of $\Ch|_n.$
\end{defn}

Finally we define the automorphism group of a set; the channel symmetry groups we will study will be \emph{subgroups} of the automorphism group of the message set.
\begin{defn}
    Given a set $A$, the \emph{automorphism group} of $A$, denoted $\Aut(A),$ is the set of bijections from $A$ to itself.
\end{defn}
In the case where $A$ is finite, we have $\Aut(A)\cong S_{|A|},$ the group of permutations of $|A|$ elements.

\subsection{Channel Symmetry Groups}
Given a channel $\Ch$ over an alphabet $\XX$, we consider the subgroup of elements of $\Aut(\XX^*)$ which ``commute with $\Ch$.'' 
\begin{defn}\label{defn:channel-group}
Given a channel $\Ch:\XX^*\times \Omega\to \YY^*$, we let the \emph{channel group} of $\Ch$ be defined as 
\begin{align*}
    \GG_{\Ch} = \{g\in \Aut(\XX^*) : \exists h\in \Aut&(\YY^*) ,\; (h\Ch g) x \eqdist \Ch x,\\
    & |g(x)|=|x| \;\forall x\in \XX^*\},
\end{align*}
where $\eqdist$ denotes equality in distribution, and by the conjugation $(h\Ch g)x$ we mean the random variable $\Omega \ni \omega \mapsto h(\Ch(g(x), \omega)).$ 
\end{defn}
\begin{rem}
In the case where $\XX=\YY,$ in many channels of interest one can fix $h=g^{-1}$ without loss of generality. In that case $\GG_{\Ch}$ is \emph{exactly} the set of bijections that commute (in the sense of equality in distribution) with $\Ch.$
\end{rem}
The above definition is a special case of the one given in Section VI.A of \cite{polyanskiy2012}; our definition differs from that one in the following two ways:
\begin{enumerate}
    \item In \cite{polyanskiy2012}, the domain and range of the channels in question are not required to have a product-like structure, or even to be countable.
    \item In \cite{polyanskiy2012}, the group elements are not required to preserve string length.
\end{enumerate}
Both of these assumptions hold automatically for many channels and symmetries of interest. The first one will simplify our proofs; the second one will allow us to have a well-defined notion of how the group ``grows'' with the blocklength: 
\begin{defn}\label{defn:group-size}
For each $n,$ consider the restrictions of the elements of $\GG_{\Ch}$ to $\{0,1\}^n$, to obtain a group $\GG_{\Ch}^n\subseteq \Aut(\{0,1\}^n)$ with a (not necessarily canonical) inclusion $\GG_{\Ch}^n\injectsto \GG_{\Ch}$. We say that the size of the channel group $\GG_{\Ch}$ at blocklength $n$ is $|\GG_{\Ch}^n|.$
\end{defn}

It's an elementary exercise to check that \ref{defn:channel-group} indeed defines a group in the formal sense; we include a proof in the appendix for completeness:
\begin{lem}\label{lem:channel-group-is-group}
The channel group $\GG_{\Ch}$ is a group.
\end{lem}

Given the channel group $\GG_{\Ch},$ we can define the natural induced notion of equivalence between strings.
\begin{defn}
Given two strings $x,y\in \XX^*$ and a channel $\Ch$ over the alphabet $\XX^*,$ we say $x$ and $y$ are \emph{equivalent with respect to }$\Ch$, and write $x\sim y,$ if there exists $g\in \GG_{\Ch}$ such that $gx=y.$  We then define the equivalence class of $x\in \XX^*$ as $[x]=\{x'\in \XX^*: x'\sim x\}.$
\end{defn}
It's again an elementary exercise to check that this defines an equivalence relation in the formal sense, and hence partitions the space of messages $\XX^*$ into a new set of disjoint equivalent classes (or \emph{orbits}) $\quot{\XX^*}{\sim}= \{[x]:x\in \XX^*\},$ called the \emph{quotient space} of $\XX^*$ by $\sim.$ We remark that, given a code $\CC = \{\CC_n\}_{n\in \N}$ for $\Ch$, where $\CC_n\subseteq \XX^n$, and a group element $g\in \GG_{\Ch},$ we can define a new code $g\CC = \{g\CC_n\}_{n\in \N},$ where by $g\CC_n$ we mean $\{g c : c\in \CC_n\}.$ The code $g\CC$ can be decoded on the channel $\Ch$ by applying $g^{-1}$ to the received message, and then using any decoder for $\CC$; by the definition of $\GG_{\Ch},$ this new decoder will have the same probability of error as the decoder for $\CC.$

\section{Invariance Theorem}\label{sec:general-results}
As a simple general result, we show that, for a wide class of channels, the distribution over the input strings that achieves the maximum mutual information between input and output is a fixed point of the action of the channel group. As was mentioned, this is a similar flavor of result to Theorem 20 in \cite{polyanskiy2012}, but in the setting of channel capacity instead of finite blocklength performance minimax bounds. As we illustrate in Section~\ref{sec:examples}, the conditions of the following theorem apply broadly.
\begin{thm}\label{thm:max-mutual-inf}
Let $\Ch$ be a channel. Suppose the $n$th transition matrix $M_n$ of $\Ch$ exists and is full-rank. If we have
\[\label{eq:variational-problem}
\DD \in \argmax_{\DD'}I(X';Y'), \tag{VP}
\]
where $X'\sim \DD'$ and $Y'=\Ch X'$, then $gX \eqdist X$ for $X\sim \DD$ for all $g \in \GG_{\Ch}$. Above, the maximum is taken over all probability distributions $\DD'$ supported on $\XX^n.$
\end{thm}
We relegate the proof to the appendix. We remark that, by Shannon's Theorem \cite{shannon}, for memoryless channels $\Ch$, if $\DD$ is the mutual-information-maximizing distribution from the theorem above and $X\sim \DD,Y=\Ch X$, then $\frac{1}{n}I(X;Y)$ is the capacity of $\Ch,$ for every $n.$ Even for non-memoryless channels like the deletion channel or other synchronization channels, Dobrushin \cite{generalized-shannons-thm} proved that the capacity is given by the limit of $\frac{1}{n}I(X;Y)$ as $n\to\infty.$ Informally, Theorem \ref{thm:max-mutual-inf} shows that, when looking for the distribution that achieves capacity, we can restrict attention to the distributions which ``respect the symmetries of the channel.'' The following corollary, which follows immediately from Theorem~\ref{thm:max-mutual-inf}, makes this more concrete.
\begin{cor}\label{cor:equiv-classes}
If $\DD$ is the mutual-information-maximizing distribution of Theorem~\ref{thm:max-mutual-inf}, then $\DD$ is uniform when restricted to the subsets of $\XX^n$ that are equivalence classes with respect to $\Ch.$ 
\end{cor}
In other words, if $x,y\in \XX^n, x\sim y$ with respect to $\Ch$, and $X\sim \DD,$ then $\pr(X=x)=\pr(X=y).$ Hence maximizing the mutual information over all distributions in $\XX^n$ is equivalent to maximizing it over the smaller quotient space $\quot{\XX^n}{\sim}.$ In particular, if we have a sufficiently large channel group, our variational problem \ref{eq:variational-problem} can reduce to a polynomial, or even constant number of variables, as it's easy to see occurs for memoryless channels \cite{polyanskiy2012}.

\section{Examples and Conjecture}\label{sec:examples}
In this section, we compare the examples of the binary symmetric channel and binary deletion channel, we give our conjecture, and we prove a simple special case. The case of the BSC was treated in \cite{polyanskiy2012}, but here we re-derive what we need for completeness and to introduce the ideas that will carry over to the case of the BDC.
\subsection{Binary Symmetric Channel \cite{polyanskiy2012}}

In order for Theorem~\ref{thm:max-mutual-inf} to apply, we need to show that the $\BSC$'s $n$th transition matrix $M_n$ is full-rank.

\begin{lem}\label{lem:bsc-full-rank}
For every $n$ and $p\in (0,1/2)$, the $n$th transition matrix of the $\BSC_p$ is full-rank.
\end{lem}
\begin{proof}
This follows automatically from the fact that the transition matrix for a single bit $M_1$ is clearly full-rank, and that $M_n = M_1^{\otimes n}$ by the memoryless property. 
\end{proof}
We now give two families of examples of elements in $\GG_{\BSC}.$ These examples were already given by Polyanskiy \cite{polyanskiy2012}. In each case, we only specify the action of the group element on strings of a particular length $n.$ A general group element may be formed by any choice a fixed-length transformation per string length.
\begin{enumerate}
    \item\label{perm} \emph{Any permutation of the indices.} This is in fact a subgroup of the channel group of any memoryless channel, as is easily verified.
    \item\label{trans} \emph{Translation by any element.} Fix an element $x \in \{0,1\}^n$ and consider the transformation $y\mapsto y+x,$ where addition is elementwise and mod 2. This is clearly an element of $\GG_{\BSC}$ with inverse equal to itself.
\end{enumerate}
While there may be other families of transformations, just the transformations of type (\ref{trans}) suffice to show that any two strings of the same length are equivalent with respect to the $\BSC.$ Namely, for $y,z\in \{0,1\}^n$, letting $x = z-y,$ the map $w\mapsto w+x$ maps $y$ to $z$, and hence $y\sim z.$ Then $[x]=\{0,1\}^n$ for every $x\in \{0,1\}^n$ and from Corollary~\ref{cor:equiv-classes} we recover the classical result that the maximum mutual information in the $\BSC$ is achieved by a uniform distribution over the input. Finally, we note that, just by considering transformations of type (\ref{perm}), $\GG_{\BSC}$ grows at least like $n!$, in the sense of Definition~\ref{defn:group-size}, and hence grows faster than any exponential function. This will be in stark contrast with the case of the $\BDC$, which we now consider.

\subsection{Binary Deletion Channel}

We check the applicability of Theorem~\ref{thm:max-mutual-inf} by showing that the $n$th transition matrix of the $\BDC$ is full-rank.
\begin{lem}\label{lem:bdc-full-rank}
For every $n$ and $p\in (0,1)$, the $n$th transition matrix of the $\BDC_p$ is full-rank.
\end{lem}
\begin{proof}
For $x \in \{0,1\}^n,$ let $\vec{x} \in \R^{\{0,1\}^n}$ be the probability vector with a $1$ in the coordinate corresponding to $x$, and zeros in all other coordinates. It suffices to show that $\{M_n \vec{x}\}_{x\in \{0,1\}^n}$ are linearly independent. But note that, for $y\in \{0,1\}^n,$ the $y$th coordinate of $M_n\vec{x}$ is nonzero if and only if $x=y$, and hence they are clearly linearly independent.
\end{proof}
As before, we now give two examples of elements of $\GG_{\BDC};$ it's trivial to verify that these are indeed in the channel group.
\begin{enumerate}
    \item \emph{A flip of all bits.} This is the same as translation by the all-ones string.
    \item \label{rotation} \emph{A rotation about the center.} This is the transformation $x_1x_2\dots x_n \mapsto x_nx_{n-1}\dots x_1.$ 
\end{enumerate}
Each of these operations are of order 2 (they are their own inverses), and they commute. Hence the group they generate is of size 4. Amazingly, we conjecture that these are \emph{all} the symmetries of the $\BDC$! More precisely, these symmetries certainly generate a \emph{subgroup} of $\GG_{\BDC};$ call this subgroup $\til{\GG}_{\BDC}.$ Then two strings $x$ and $y$ are equivalent with respect to $\til{\GG}_{\BDC}$ if there is $g\in \til{\GG}_{\BDC}$ (i.e. either a flip of all bits or a rotation around the center, or their composition, or the identity) such that $gx=y.$ We conjecture that any two strings $x,y\in \{0,1\}^n$ are equivalent with respect to $\GG_{\BDC}$ if and only if they are equivalent with respect to $\til{\GG}_{\BDC}.$ This in particular would imply the following concise statement.
\begin{conj}
The size of any equivalence class (orbit) of $\{0,1\}^n$ under the action of the channel group of the BDC is at most 4, independently of $n.$
\end{conj}
In other words, when searching for the distribution that achieves capacity, while in the case of the $\BSC$ symmetry suffices to solve the problem, in the case of the $\BDC$ it \emph{essentially buys us nothing}.

To justify our conjecture, we note that if two strings $x,y\in \{0,1\}^n$ are equivalent with respect to the $\BDC_d,$ then since by definition of equivalence there is a group element such that $g^{-1}\BDC_{1/2}y \eqdist \BDC_{1/2}x,$ in particular we must have $H(\BDC_{1/2}y) = H(\BDC_{1/2}x).$\footnote{Here the choice of $d=1/2$ is of course arbitrary; the channel group should be invariant to the parameter, as long as $d\in (0,1).$} In the data files available in arXiv together with this paper, we show numerically that for string lengths $n$ up to 14, the equivalence classes under $\til{\GG}_{\BDC}$ coincide exactly with the sets of strings of equal entropy when passed through the $\BDC_{1/2}.$

We now prove a simple but significant special case of our conjecture. Let $\Perm(\{0,1\}^*)$ be the subgroup of $\Aut(\{0,1\}^*)$ of all $h\in \Aut(\{0,1\}^*)$ that act on $\{0,1\}^n$ by a permutation of the indices. Specifically, for $h\in \Aut(\{0,1\}^*)$, we have $h\in \Perm(\{0,1\}^*)$ if and only if for all $n\in \N$, there exists a permutation $\pi_n:[n]\to[n]$ such that for all $x\in \{0,1\}^n$, we have
\[
h(x) = x_{\pi_n(1)}x_{\pi_n(2)}\dots x_{\pi_n(n)}.
\]
We note that in the case of the BSC (and memoryless channels more generally \cite{polyanskiy2012}), we have $\GG_{\BSC}\cap \Perm(\{0,1\}^*) = \Perm(\{0,1\}^*),$ and of course $\Perm(\{0,1\}^*)$ grows like $n!$, in the sense of Definition~\ref{defn:group-size}. By contrast, we prove the following.

\begin{prop}\label{prop:conj-special-case}
If $g\in \GG_{\BDC} \cap \Perm(\{0,1\}^*),$ then $g$ is either a rotation about the center, in the sense of (\ref{rotation}), or the identity.
\end{prop}
Before proving the proposition, we prove a useful lemma. We use the following notation: for $i\in [n]$, we let $\Ne_n(i) := \{i-1,i+1\}\cap[n]$ be the indices neighboring $i.$
\begin{lem}\label{lem:loc-prop}
A symmetry $h\in \Perm(\{0,1\}^*)$ is either a rotation about the center or the identity if and only if we have the following local property for all $n$:
\[\label{eq:LP}
j\in \Ne_n(i) \iff \pi_n(j)\in \Ne_n(\pi_n(i)) \quad \forall i\neq j\in [n]\tag{LP}.
\]
\end{lem}
\begin{proof}
Clearly if $h$ is a rotation about the center or the identity, then it satisfies (\ref{eq:LP}) for all $n$. For the other direction, we proceed by induction on the blocklength $n$. For $n=1,2,$ our thesis holds trivially. For $n\geq 3,$ property (\ref{eq:LP}) implies the following on the first $n-1$ indices:
\[\label{eq:LP'}
j\in \Ne_n(i) \iff \pi_n(j)\in \Ne_n(\pi_n(i)) \quad \forall i\neq j\in [n-1], \tag{LP'}
\]
which implies that $\pi_n([n-1])$ is contiguous, and hence equal to $[n-1]$ or $[n]\setminus\{1\}.$ After a shift by $-1$ in the latter case, we can identify the restriction $\pi|_{[n-1]}$ with a permutation $\til{\pi}_{n-1}:[n-1]\to[n-1].$ Property (\ref{eq:LP'}) for $\pi_n$ clearly implies property (\ref{eq:LP}) for $\til{\pi}_{n-1}$. Hence by inductive hypothesis, $\til{\pi}_{n-1}$ is either the identity or a rotation about the center. Going back to $\pi_n,$ we have the following two cases:
\begin{enumerate}
    \item If $\pi_n([n-1]) = [n-1],$ then $\til{\pi}_{n-1}$ cannot be a rotation about the center, or else condition (\ref{eq:LP}) is violated at $i=n.$ Hence $\til{\pi}_{n-1} = \pi_n|_{[n-1]}$ is the identity and $\pi(n)=n,$ so $\pi_n$ is the identity as well.
    \item  If $\pi_n([n-1]) = [n]\setminus\{1\},$ then $\til{\pi}_{n-1}$ cannot be the identity, or else condition (\ref{eq:LP}) is violated at $i=n.$ Hence $\til{\pi}_{n-1} = \pi_n|_{[n-1]}-1$ is a rotation about the center, and since $\pi_n(n)=1,$ so is $\pi_n.$
\end{enumerate}
This completes the proof.
\end{proof}
Now we can prove the proposition.
\begin{proof}[Proof of Proposition~\ref{prop:conj-special-case}]
Let $g\in \GG_{\BDC}\cap \Perm(\{0,1\}^*),$ and let $\pi_n:[n]\to [n]$ be the associated index permutations at each blocklength. Suppose for contradiction that $g$ is neither the identity nor a rotation about the center. Then by Lemma~\ref{lem:loc-prop}, we must have $n$ and $i\neq j\in [n]$ such that $i$ and $j$ are neighbors but $\pi_n(i)$ and $\pi_n(j)$ are not. Assume without loss that $j=i+1$ and (by composing $\pi_n$ with a rotation about the center on the \emph{right} if needed) that $\pi_n(i) + 1<\pi_n(i+1).$ Let $k\in [n]$ be such that $\pi_n(i) < \pi_n(k)< \pi_n(i+1),$ and assume, again without loss (by composing $\pi_n$ with a rotation about the center on the \emph{left} if needed), that $k > i+1.$ Consider the string $x$ which is all-ones up to index $i+1,$ inclusive, and all-zeros afterwards (see Figure~\ref{fig:x-string}). We claim that $|\supp(\BDC_d x)| \neq |\supp(\BDC_d gx)|$ for any $d$.\footnote{Here and below $\supp$ denotes the support of the distribution.} This of course automatically implies that there can be no $h\in \Aut(\YY^*)$ such that $\BDC_d x \eqdist h \BDC_d gx,$ yielding a contradiction of $g\in \GG_{\BDC}$ and proving the proposition. To prove the claim, note first that, by a simple calculation, we have $|\supp(\BDC_d x)| = (i+2)(n-i)$: there is exactly one possible subsequence of $x$ for each valid choice of a number of ones and a number of zeros in the output. Moreover since $g$ preserves string weight, $x$ and $g(x)$ have the same total number of ones and zeros. Hence $|\supp(\BDC_d g(x))| \geq (i+2)(n-i)$ by the same argument: there is a (now not necessarily unique) subsequence of $g(x)$ for each valid choice of a number of ones and a number of zeros in the output. To show that the inequality is strict, it suffices to exhibit two indices $s\neq t$ such that $g(x)_s=g(x)_t=1$, and such that deleting index $s$ (and nothing else) from $g(x)$ will produce a different output than deleting index $t$. A moment of thought reveals that $s=\pi_n(i)$ and $t=\pi_n(i+1)$ works, since they don't lie in the same contiguous block of ones by construction (see Figure~\ref{fig:x-string}).
\begin{figure}
    \centering
    \includegraphics[width=0.9\linewidth]{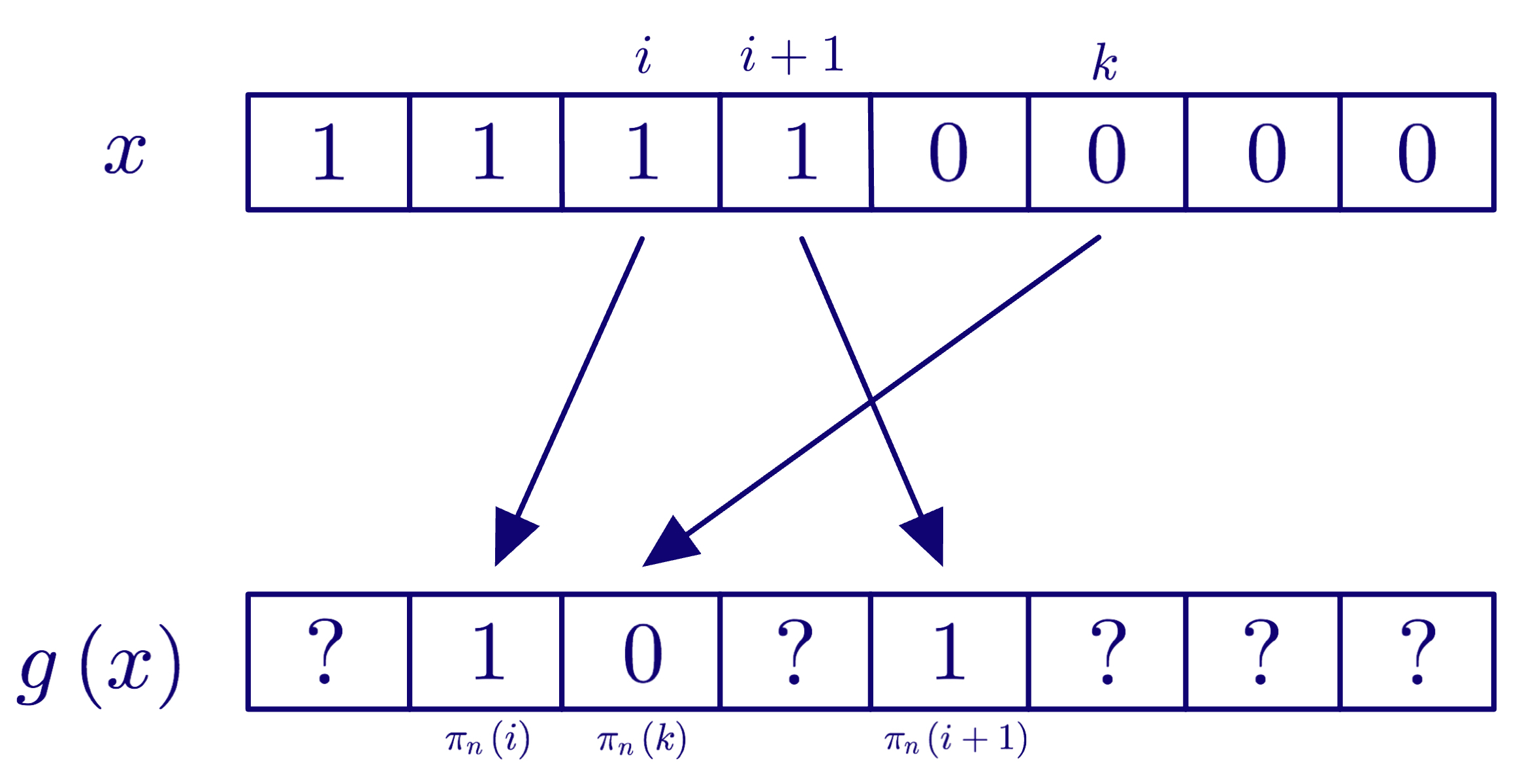}
    \caption{An example of the string $x\in\{0,1\}^n$ used in our proof, and the group element $g$ acting on it via the permutation $\pi_n$ of its indices. The question marks denote bits of $g(x)$ which we cannot determine just with our assumptions on $g.$}
    \label{fig:x-string}
\end{figure}

\end{proof}

\section{Acknowledgements}
The author would like thank Ray Li, Mary Wootters and Marco Mondelli for illuminating discussions and invaluable feedback on early versions of the ideas presented in this paper. We also thank anonymous reviewers for their comments, which led to the simplification of various proofs, and for pointing our some errors and reference omissions (mentioned above) in a prior version of this paper.




\bibliographystyle{IEEEtran}
\bibliography{IEEEabrv,refs}

\section{Appendix}
\begin{proof}[Proof of Lemma~\ref{lem:channel-group-is-group}]
We verify the axioms:

\begin{enumerate}
    \item \emph{Inclusion of 1.} We trivially have $1 \in \GG_{\Ch}.$
    \item \emph{Closedness under inverses.} Suppose $g\in \GG_{\Ch}$. Clearly $g^{-1}\in \Aut(\XX^*)$ and $g^{-1}$ preserves string length. Moreover by assumption there exists $h\in \Aut(\YY^*)$ such that $h\Ch gx \eqdist \Ch x$ for all $x \in \XX^*.$ Multiplying by $h^{-1}$ on the left and letting $x'=gx,$ we have $\Ch x' \eqdist h^{-1}\Ch g^{-1}x'$. But $\{gx : x\in \XX^*\} = \XX^*$ since $g \in \Aut(\XX^*),$ so $g^{-1} \in \GG_{\Ch}.$
    \item \emph{Closedness under multiplication.} Let $g_1,g_2 \in\GG_{\Ch}.$ Clearly we have $g_1g_2 \in \Aut(\XX^*),$ and $g_1g_2$ preserves string length. Moreover, letting $h_1,h_2$ be such that $h_i\Ch g_i x \eqdist \Ch x$ for all $x\in \XX^*, i=1,2,$ we have
    \begin{align*}
        h_2h_1 \Ch g_1g_2 x &=h_2 (h_1\Ch g_1)(g_2x) \\
        &\eqdist h_2 \Ch g_2x \\
        &\eqdist \Ch x,
    \end{align*}
    so $g_1g_2 \in G_\DD(\Ch)$.
\end{enumerate}
\end{proof}

\begin{proof}[Proof of Theorem \ref{thm:max-mutual-inf}]
Fix a channel $\Ch$ over an alphabet $\XX$ and a group element $g\in \GG_{\Ch}.$ Suppose $\DD$ is the mutual-information-maximizing distribution of the statement of the theorem, $X\sim \DD$ and $Y=\Ch X.$ Then let $X' = gX$ and $Y'=\Ch X'.$ We will prove two claims: (1) that $I(X;Y) = I(X';Y'),$ and (2) that, under the assumptions of the theorem, the mutual information $I(X;Y)$ is maximized by a unique distribution of $X.$ The combination of these two claims yields $X'\eqdist X,$ proving the theorem.

For the first claim, we show the stronger statement that, for $h$ as in Definition~\ref{defn:channel-group}, we have $(g^{-1}X',hY')\eqdist(X,Y)$; since $\XX\times\YY\ni(x,y)\mapsto (g^{-1}x, hy)$ is a bijection, this then immediately gives (1). Indeed, we have
\begin{align*}
    \pr((g^{-1}X', &hY') = (x,y)) \\
    &=\pr(g^{-1} g X = x,h\Ch g X  = y) \\
    &= \pr( X = x,h\Ch g X  = y) \\
    &= \pr( X = x)\pr(h\Ch g X  = y| X = x) \\
    &= \pr( X = x)\pr(h\Ch g x  = y) \\
    &= \pr( X = x)\pr(\Ch x  = y) \\
    &= \pr( X = x)\pr(\Ch X  = y| X = x) \\
    &= \pr(X = x, \Ch X  = y) \\
    &= \pr((X,Y) =(x,y)),
\end{align*}
as desired.

For the second claim, if $X\sim \DD_X$, it suffices to show that the function $\DD_X \mapsto I(X;Y)$ is strictly concave.\footnote{Note that the full rank assumption on the transition matrix is necessary here: while entropy is always a strictly concave function of the distribution, mutual information is only a concave (and not necessarily strictly concave) function of the input distribution.} But writing
\[
I(X;Y) = H(Y) - H(Y|X),
\]
the second term is a linear function of $\DD_X$ (it can be written $H(Y|X) = \E_{x\sim \DD_X} H(Y|X=x)$), and the first is a strictly concave function (the entropy) of the distribution of $Y$ (call it $\DD_Y$). Now we may obtain $\DD_Y$ as a linear function of the distribution of $X$, i.e. $\DD_Y = M_n \DD_X$ (where $M_n$ is the $n$th transition matrix of $\Ch$), and since $M_n$ is full-rank, $\DD_X\mapsto \DD_Y$ is an injective linear function, hence $\DD_X\mapsto H(Y)$ is strictly concave. This proves (2) and hence the theorem.
\end{proof}

\end{document}